\documentclass[a4paper,twocolumn,accepted=2022-06-02]{quantumarticle}

\pdfoutput=1
\usepackage[utf8]{inputenc}
\usepackage[english]{babel}
\usepackage[T1]{fontenc}
\usepackage{amssymb}
\usepackage{graphicx}
\usepackage{graphics}
\usepackage{amsmath}
\usepackage{amsthm}
\usepackage{color}
\usepackage{hyperref}

\usepackage{ulem}

\newcommand{\bea}{\begin{eqnarray}}
\newcommand{\eea}{\end{eqnarray}}

\def\bi{\begin{itemize}}
\def\ei{\end{itemize}}
\def\bc{\begin{center}}
\def\ec{\end{center}}

\def\C{\hbox{$\mit I$\kern-.7em$\mit C$}}
\def\R{\hbox{$\mit I$\kern-.6em$\mit R$}}

\newcommand{\one}{\mbox{$1 \hspace{-1.0mm}  {\bf l}$}}

\def\conv{\textrm{conv}}

\newtheorem{theorem}{Theorem}

\newtheorem{lemma}{Lemma}

\newcommand{\uno}{\tilde{\mbox{$1 \hspace{-1.0mm}  {\bf l}$}}}

\begin{document}

\title{Genuine multipartite entanglement of quantum states in the multiple-copy scenario}
\author{Carlos Palazuelos}
\affiliation{Departamento de An\'alisis Matem\'atico y Matem\'atica Aplicada, Universidad Complutense de Madrid, E-28040 Madrid, Spain}
\affiliation{Instituto de Ciencias Matem\'aticas, E-28049 Madrid, Spain}
\author{Julio I. de Vicente}
\affiliation{Departamento de Matem\'aticas, Universidad Carlos III de
Madrid, E-28911, Legan\'es (Madrid), Spain}

\begin{abstract}
Genuine multipartite entanglement (GME) is considered a powerful form of entanglement since it corresponds to those states that are not biseparable, i.e.\ a mixture of partially separable states across different bipartitions of the parties. In this work we study this phenomenon in the multiple-copy regime, where many perfect copies of a given state can be produced and controlled. In this scenario the above definition leads to subtle intricacies as biseparable states can be GME-activatable, i.e.\ several copies of a biseparable state can display GME. We show that the set of GME-activatable states admits a simple characterization: a state is GME-activatable if and only if it is not partially separable across one bipartition of the parties. This leads to the second question of whether there is a general upper bound in the number of copies that needs to be considered in order to observe the activation of GME, which we answer in the negative. In particular, by providing an explicit construction, we prove that for any number of parties and any number $k\in\mathbb{N}$ there exist GME-activatable multipartite states of fixed (i.e.\ independent of $k$) local dimensions such that $k$ copies of them remain biseparable.

\end{abstract}

\maketitle

\section{Introduction}

Entanglement in quantum many-body systems plays a key role in quantum information science and in the analysis of condensed matter physics. In fact, multipartite entangled states are a crucial resource in many applications of quantum technologies such as quantum sensing \cite{sensing,sensing2}, secure quantum communication \cite{qcka} or certain schemes for quantum computation like the one-way quantum computer \cite{1wayqc}. This has given rise to an extensive body of work both experimental and theoretical (see e.g.\ the reviews \cite{review1,review2}) aimed, on the one hand, at the preparation and control of entangled states shared by an increasing number of parties and, on the other hand, at the certification, classification and quantification of multipartite entanglement. This last task is plagued with several serious difficulties. For instance, the question of whether a given state is entangled or not has been proven to be computationally hard even in the bipartite scenario \cite{gurvits}, multipartite entanglement in pure states manifests itself in infinitely many inequivalent ways \cite{slocc1,slocc2,slocc3} and its manipulation under local operations and classical communication (LOCC) is severely constrained \cite{locc1,locc2,locc3}.

A natural and relevant question in this context is: which form of entanglement in $n$-partite quantum systems qualifies as truly multipartite entanglement spread over all $n$ subsystems? Entangled states defined to be those that are not fully separable cannot do the job. This is because an $n$-partite state such as $|\psi\rangle_{12}\otimes|0\cdots0\rangle_{3\cdots n}$, where $|\psi\rangle$ is a bipartite entangled state, certainly contains some entanglement and is not fully separable; yet, it is not entangled among all its constituents but only among a strict subset. Thus, states that are not fully separable but separable with respect to a bipartition of the parties are referred to as partially separable states. In order to cope with this fact, the notion of genuine multipartite entangled (GME) states was introduced \cite{GMEdef} as those that cannot be written as a mixture of partially separable states over different bipartitions (i.e.\ as those that are not biseparable). The question of whether a given $n$-partite state is GME or not is a highly non-trivial one (while deciding partial separability is already hard, the situation here gets even more complicated because there are biseparable states which are not partially separable; see e.g.\ \cite{bsneqps}); however, it has been thoroughly studied and used in experimental implementations as a benchmark \cite{review1,review2}. Indeed, most relevant states for applications--like graph states--are GME and the set of biseparable states is closed under LOCC; thus, this classification is well-defined within the standard paradigm of state manipulation under LOCC and biseparable states are then fundamentally limited for many applications. Actually, it has been shown that GME is in general necessary to achieve maximum sensitivity in quantum metrology \cite{gmesensing1,gmesensing2} and to establish a multipartite secret key \cite{GMEkey}.

Nevertheless, all these considerations only involve the single-copy regime. It is at least in principle conceivable that if an experimental implementation outputs an $n$-partite state $\rho$, repeating this procedure endorsed with a quantum memory would lead to the preparation of multiple identically prepared copies of the state, i.e.\ $\rho^{\otimes k}$. Remarkably, the set of biseparable states is not tensor stable: there exist biseparable states such that  several copies of them become GME \cite{purificationGME}. Thus, we say that a state $\rho$ is GME-activatable if $\rho^{\otimes k}$ is GME for some $k\in\mathbb{N}$ (notice that the closedness of the set of biseparable states under LOCC implies that if $\rho^{\otimes k_0}$ is GME, then $\rho^{\otimes k}$ is GME for all $k\geq k_0$). Hence, the ability to prepare and control many copies of a biseparable state can pass the GME test and offers the possibility at least in principle to obtain useful states for applications (possibly conditioned on further LOCC postprocessing \cite{purificationGME}). For instance, certain biseparable states have been shown to lead to secure multipartite secret key exploiting their GME-activability \cite{biseparableqcka}. As another example, all connected networks made out by sharing bipartite pure entangled states happen to be GME and genuine multipartite non-locality can be extracted out of them \cite{cavalcanti,networkGMNL}. From a more fundamental point of view, although examples of activation and superactivation are not infrequent in quantum information theory (they appear e.g.\ when considering which bipartite quantum states are non-local \cite{palazuelos_super-activation_2012} or when assessing which quantum channels have a non-zero capacity \cite{superactivationchannels}) and they constitute an intriguing feature of quantum theory deeply related to the intricacies of the tensor product structure, this property of biseparable states might lead to question whether GME actually certifies truly multipartite entanglement if multiple-copy manipulation is available. In fact, for similar reasons a different and more restrictive multipartite notion of entanglement has been recently proposed in \cite{GNME1}.

Given these considerations, a clear question emerges: what is the set of GME-activatable states? Reference \cite{GMEactivation} has recently considered this question for particular classes of states and conjectured that a state is not GME-activatable if and only if it is partially separable. The first main result of this paper is a proof of this conjecture, hence characterizing the set of GME-activatable states in full generality. On the one hand, this implies that if we cannot trust an experimentalist not to store copies, then the only thing certified by GME is non-partial separability, a property that can be certified by arguably simpler tests. This stresses the limitations of GME as a notion of truly multipartite entanglement already pointed out in \cite{GNME1}. On the other hand, our result can be seen as a means to produce useful multipartite entanglement by mixing partially separable states and subsequent LOCC processing on multiple copies. While producing and controlling multiple copies of a state is clearly experimentally more demanding \cite{copies}, this is within current technological reach in certain experiments for a moderate number of copies \cite{GMEactivation}. This motivates a second question: given that $\rho$ is GME-activatable, what is the smallest $k\in\mathbb{N}$ such that $\rho^{\otimes k}$ is GME? This question has also been considered in \cite{GMEactivation}, where it is shown that there exist GME-activatable states such that two copies of them are still not GME. It was also conjectured therein that there is no upper bound on the number of copies of a GME-activatable state that need to be considered in general in order to display GME. Our second main result is again a proof of this conjecture. We show that for any number of parties $n$ ($n\geq3$) and any $k\in\mathbb{N}$, there exists an $n$-partite state $\rho$ that is GME-activatable but $\rho^{\otimes k}$ is biseparable. Thus, the convergence of the hierarchy of sets of states which are not GME-activatable with $k$ copies to the set of partially entangled states only occurs in the limit $k\to\infty$. Therefore, the preparation of GME states by mixing partially separable states can require control over an arbitrarily large number of copies.

\section{Notation and definitions}

In this paper $n$ will denote the number of parties and $[n]:=\{1,2,\ldots,n\}$. Thus, any given $n$-partite quantum system will be associated with the Hilbert space $H=\bigotimes_{i=1}^nH_i$ where $H_i=\mathbb{C}^{d_i}$ for natural numbers $d_i\geq2$ for any party $i$. In particular, all Hilbert spaces will be assumed to be finite dimensional in this work. The set of density matrices will be denoted by $D(H)\subset B(H)$. A pure state $|\psi\rangle\in H$ is fully separable if there exists pure states $|\phi_i\rangle\in H_i$ $\forall i$ such that
\begin{equation}\label{fspure}
|\psi\rangle=\bigotimes_{i=1}^n|\phi_i\rangle,
\end{equation}
while it is partially separable if for some $M\subsetneq[n]$, $M\neq\emptyset$ (in order to ease the notation from now on we will always assume that subsets of $[n]$ are non-empty without explicitly stating it), there exist pure states $|\phi_M\rangle\in H_M:=\bigotimes_{i\in M}H_i$ and $|\phi_{\bar{M}}\rangle\in H_{\bar{M}}=\bigotimes_{i\notin M}H_i$ ($\bar{M}$ denotes the complement of $M$ in $[n]$) such that
\begin{equation}\label{pspure}
|\psi\rangle=|\phi_M\rangle\otimes|\phi_{\bar{M}}\rangle.
\end{equation}

These definitions are then extended to arbitrary states in $D(H)$ by taking convex hulls (which we denote by $\conv$). The set of fully separable density matrices, $FS(H)$, is given by the convex hull of all pure density matrices $|\psi\rangle\langle\psi|$ such that $|\psi\rangle$ satisfies Eq.\ (\ref{fspure}). The set of partially separable density matrices in the bipartition $M|\bar{M}$, $S_M(H)$, is given by the convex hull of all pure density matrices $|\psi\rangle\langle\psi|$ such that $|\psi\rangle$ satisfies Eq.\ (\ref{pspure}) for that given $M$. Finally, the set of partially separable density matrices is $S(H)=\cup_{M\subsetneq[n]}S_M(H)$ and the set of biseparable density matrices is $BS(H)=\conv\{S(H)\}$. For any $M\subsetneq[n]$ we have that $FS(H)\subset S_M(H)\subset S(H)\subset BS(H)$ and a state is GME if it is not biseparable. Thus, $\rho\in D(H)$ is in $S_M(H)$ if and only if (iff)
\begin{equation}\label{ps}
\rho=\sum_jq_j\sigma^{(j)}_M\otimes\tau_{\bar{M}}^{(j)},
\end{equation}
where for all $j$ $\sigma^{(j)}_M\in D(H_M)$ and $\tau^{(j)}_{\bar{M}}\in D(H_{\bar{M}})$ and $\{q_j\}$ are convex weights, i.e.\ $\sum_jq_j=1$ and $q_j\geq0$ $\forall j$. It holds as well that $\rho\in D(H)$ is in $BS(H)$ iff
\begin{equation}\label{bs}
\rho=\sum_{M\subsetneq[n]}p_M\chi_M,
\end{equation}
where $\{p_M\}$ are convex weights and $\chi_M\in S_M(H)$ $\forall M\subsetneq[n]$. Obviously, the sum in $M$ above needs only to run over those subsets that give rise to different bipartitions $M|\bar{M}$. Hence, without loss of generality, we can take that the sum contains at most $2^{n-1}-1$ terms.

Notice that in the bipartite case $H=H_1\otimes H_2$ all the above sets boil down to the same set: $FS(H)=S_1(H)=S_2(H)=S(H)=BS(H)$, i.e.\ the set of separable states, which we will denote for simplicity throughout the paper by $S(H)$. Here, non-separable states are referred to as entangled states. In this case, for any $\rho\in D(H=H_1\otimes H_2)$ we will make use of the following quantity
\begin{equation}\label{t}
T(\rho):=\frac{1}{2}\min_{\sigma\in S(H)}||\rho-\sigma||,
\end{equation}
where $||\cdot||$ stands for the trace norm (the minimum in this expression is justified by the compactness of $S(H)$). Notice that for any $\rho\in D(H)$ it holds that $0\leq T(\rho)\leq1$ and all separable states attain the lower bound. Actually, the trace-norm distance to the set of separable states is an entanglement measure in the sense that it cannot increase under deterministic LOCC transformations \cite{vedralplenio}. Furthermore, $T(\rho)$ is convex as for any $p\in[0,1]$ and any $\rho_1,\rho_2\in D(H)$ letting $\sigma_1,\sigma_2\in S(H)$ be the closest separable states to $\rho_1$ and $\rho_2$ respectively, we have that
\begin{align}
&T(p\rho_1+(1-p)\rho_2)\nonumber\\&\leq\frac{1}{2}||p\rho_1+(1-p)\rho_2-p\sigma_1-(1-p)\sigma_2||\nonumber\\
&\leq\frac{1}{2}\left(p||\rho_1-\sigma_1||+(1-p)||\rho_2-\sigma_2||\right)\nonumber\\
&=pT(\rho_1)+(1-p)T(\rho_2).
\end{align}
Going back to the multipartite case $H=\bigotimes_{i=1}^nH_i$, we will denote by $T_{M|\bar{M}}(\rho)$ the value of this quantity when we think of $\rho\in D(H)$ as a bipartite state in $D(H_M\otimes H_{\bar{M}})$, i.e. $T_{M|\bar{M}}(\rho):=\frac{1}{2}\min_{\sigma\in S_M(H)}||\rho-\sigma||$.

When we consider $k\in\mathbb{N}$ copies of an $n$-partite state, it should be understood that each copy has Hilbert space $H^{(j)}=\bigotimes_{i=1}^nH_i^{(j)}$ with $H_i^{(j)}=\mathbb{C}^{d_i}$ for $j=1,2,\ldots,k$. The multicopy state is then associated to the Hilbert space $H=\bigotimes_{i=1}^nH_i$ with $H_i=\bigotimes_{j=1}^kH_i^{(j)}$, to which all the above definitions apply. As stated in the introduction, an $n$-partite state $\rho$ is GME-activatable if there exists $k\in\mathbb{N}$ such that $\rho^{\otimes k}\notin BS(H)$.

\section{Characterization of the set of GME-activatable states}

The main result of this section is given in Theorem \ref{th1} below. We will use the following lemma in \cite{beigi} (although this is not explicitly stated there, the reader can readily check that this is the argument established in this reference in order to prove Theorem 1 therein).

\begin{lemma}(\cite{beigi}) If $\rho\in D(H_1\otimes H_2)$ is entangled, then
$$\lim_{n\to\infty}T(\rho^{\otimes n})=1.$$
\end{lemma}

\begin{theorem}\label{th1}
Let $n\in\mathbb{N}$ ($n\geq2$) and $H=\bigotimes_{i=1}^nH_i$. Then, $\rho\in D(H)$ is GME-activatable iff $\rho\notin S(H)$
\end{theorem}
\begin{proof}
Clearly, $S(H)$ is tensor stable; therefore, we only need to prove that if $\rho\notin S(H)$, then $\rho$ is GME-activatable. For this, we proceed in two steps. We first establish an upper bound to the sum of $T_{M|\bar{M}}$ over all bipartitions for biseparable states and then we show that it is violated by sufficiently many copies of any state that is not partially separable. If $\sigma$ is a biseparable state, then (cf.\ Eq.\ (\ref{bs}))
\begin{equation}
\sigma=\sum_{M\subsetneq[n]}p_M\chi_M,
\end{equation}
where $\{p_M\}$ are convex weights, for each $M$ $\chi_M\in S_M(H)$ and, here and in the remainder of the proof, the sum runs over a choice of all subsets that give rise to different bipartitions $M|\bar{M}$ (thus, as mentioned in the previous section, the sum above has in particular at most $2^{n-1}-1$ terms). If we consider a particular subset $M$, the convexity of $T$ implies then that
\begin{equation}
T_{M|\bar{M}}(\sigma)\leq1-p_M.
\end{equation}
Hence,
\begin{equation}
\sum_{M}T_{M|\bar{M}}(\sigma)\leq2^{n-1}-2
\end{equation}
must hold for every biseparable state $\sigma$. We now see that this inequality is violated for sufficiently many copies of a state $\rho\notin S(H)$, i.e.\ such that $\rho\notin S_M(H)$ $\forall M\subsetneq[n]$. By Lemma 1, this premise implies that for any non-empty strict subset $M$ of $[n]$ and $\forall\epsilon>0$, there exists $k_M(\epsilon)\in\mathbb{N}$ such that for all $k\geq k_M(\epsilon)$
\begin{equation}
T_{M|\bar{M}}(\rho^{\otimes k})\geq1-\epsilon.
\end{equation}
Thus, by choosing $k\in\mathbb{N}$ large enough so that the equation above holds for all $M$ and some
\begin{equation}
\epsilon<\frac{1}{2^{n-1}-1},
\end{equation}
it follows that
\begin{equation}
\sum_{M}T_{M|\bar{M}}(\rho^{\otimes k})>2^{n-1}-2,
\end{equation}
proving that $\rho$ must be GME-activatable.
\end{proof}

\section{Activation of GME requires arbitrarily many copies}

Theorem \ref{th1} naturally leads to wonder how large the constant $k$ must be for a GME-activatable state $\rho$ to display GME. In fact, a more precise way to approach this problem is to study whether there exists a function $f(n,d)$ such that for any $n$-partite GME-activatable state $\rho$ of local dimension $d$ we have that $\rho^{\otimes k}$ is GME for $k\geq f(n,d)$. It was also conjectured in \cite{GMEactivation} that such an upper bound does not exist. The main result of this section is a proof of this conjecture.
\begin{theorem}\label{th2}
For any $n, k\in\mathbb{N}$, with $n\geq3$, there exists a $n$-partite GME-activatable state $\sigma\in D(\mathbb{C}^{2^{n-1}}\otimes  \bigotimes_{i=2}^n\mathbb{C}^{2})$ such that $\sigma^{\otimes k}$ is biseparable.
\end{theorem}

In order to prove this result, let us recall the definition of isotropic states, which are given by
\begin{equation}\label{eqisotropic1}
\rho(p)=p\phi^++(1-p)\tilde{\one},
\end{equation}
where $|\phi^+\rangle=(|00\rangle+|11\rangle)/\sqrt{2}$ is the $2$-dimensional maximally entangled state, $\phi^+=|\phi^+\rangle\langle\phi^+|$ and $\tilde{\one}=\one/4$ (we will also use the notation $\tilde{\one}$ for the normalized identity in any finite dimension). It is well-known that $\rho(p)$ is entangled iff $p>1/3$ \cite{HoHo99}.

We will also need a couple of lemmas about bipartite states. Given a Hilbert space $H=H_1\otimes H_2$, recall that $S(H)$ denotes the set of separable states.
\begin{lemma}\label{lemma 1}
Let $\Phi\in D(H)$ and $\alpha\in (0,1)$ be such that $\rho_1=\alpha\Phi+(1-\alpha)\uno\in S(H)$. Then, for any $\alpha'\in [0,\alpha)$ the state $\rho_2=\alpha'\Phi+(1-\alpha')\uno$ is in the interior of $S(H)$.
\end{lemma}
\begin{proof}
We know that the state $\uno$ is in the interior of $S(H)$ \cite{GB02}. Now, since $\rho_1\in S(H)$, \cite[Lemma 1.1.9]{Schneider} assures that $$[\uno, \rho_1):=\{\lambda \uno+ (1-\lambda)\rho_1: \lambda \in [0,1)\}$$ is contained in the interior of $S(H)$. Now, it is clear that $\rho_2\in [\uno, \rho_1)$.
\end{proof}

\begin{lemma}\label{lemma 2}
Let $\Phi\in D(H)$, $\alpha\in (0,1)$, $n\in\mathbb{N}$ and $\tilde{\rho}_1$ be the normalized state of $\rho_1=\frac{\alpha}{n}\Phi+(1-\alpha)\uno$. Then, for any $\alpha'\in [0,\alpha]$ the normalized state $\tilde{\rho_2}$ of $\rho_2=\frac{\alpha'}{n}\Phi+(1-\alpha')\uno$ verifies that $\tilde{\rho}_2\in [\uno, \tilde{\rho}_1]$. In particular, $\tilde{\rho}_1\in S(H)$ implies $\tilde{\rho}_2\in S(H)$.
\end{lemma}
\begin{proof}
It is trivial from the fact that $0\leq \alpha'\leq \alpha$ implies $$\frac{\frac{\alpha'}{n}}{\frac{\alpha'}{n}+(1-\alpha')}\leq \frac{\frac{\alpha}{n}}{\frac{\alpha}{n}+(1-\alpha)}.$$

The last assertion follows from the fact that $S(H)$ is a convex set.
\end{proof}

We are now in the position to prove Theorem 2.

\begin{proof}
The proof proceeds by explicitly constructing the required state for any given value of $k$ and $n$. The states will be pair-entangled network (PEN) states \cite{CPV21} arising in a star network where the parties share isotropic states as in (\ref{eqisotropic1}). In what follows, subindices on bipartite states will denote which parties share it, e.g.\ $\rho_{ij}(p)$ represents an isotropic state with visibility $p$ (cf.\ Eq.\ (\ref{eqisotropic1})) shared by parties $i$ and $j$ in $[n]$. The PEN states that we use in our construction are then
\begin{equation}\label{penstar1}
\sigma_n(p)=\bigotimes_{i=2}^n\rho_{1i}(p).
\end{equation}
Here, party $1$ is treated as if composed of $n-1$ qubits and each term in the tensor product addresses a different one of these qubits. Notice also that $\sigma_n(p)$ is obviously not partially separable $\forall n\geq3$ and $\forall p>1/3$, and therefore, by Theorem 1, GME-activatable (this had already been observed in \cite{CPV21} using an entanglement distillation argument). Thus, it remains to provide a value $p=p(k,n)>1/3$ for which $\sigma_n(p)^{\otimes k}$ is biseparable.

Let us write, for a general $p\in [0,1]$,
\begin{equation}
\rho(p)^{\otimes k}=f_k(p)\Phi(p,k)+(1-f_k(p))\uno^{\otimes k},
\end{equation}
where $f_k(p)=1-(1-p)^k$ and $\Phi(p,k)$ is a (possibly entangled) state.  With this notation at hand, we can rewrite $k$ copies of the state in (\ref{penstar1}), $\sigma_n(p)^{\otimes k}=\bigotimes_{i=2}^n\rho_{1i}(p)^{\otimes k}$, as
\begin{align}\label{split1}
&[f_k(p)]^{n-1}\bigotimes_{i=2}^n\Phi_{1i}(p,k)\\&\nonumber+[f_k(p)]^{n-2}(1-f_k(p))\sum_{i=2}^n\left(\bigotimes_{j\neq i}\Phi_{1j}(p,k)\right)\otimes\uno^{\otimes k}_{1i}\\&\nonumber +\cdots
\end{align}
where the omitted terms are all manifestly biseparable. Now, Eq. (\ref{split1}) can be rewritten as
\begin{widetext}
\begin{equation}\label{split2}
[f_k(p)]^{n-2}\sum_{i=2}^n\bigotimes_{j\neq i}\Phi_{1j}(p,k)\otimes \left(\frac{f_k(p)}{n-1}\Phi_{1i}(p,k)+(1-f_k(p))\uno^{\otimes k}_{1i}\right)+\cdots
\end{equation}
\end{widetext}

Hence, the proof is concluded if we find a $p=p(n,k)>1/3$ such that the state (after normalization)
$$\frac{f_k(p)}{n-1}\Phi(p,k)+(1-f_k(p))\uno^{\otimes k}$$ is separable.

Now, we know that
\begin{equation}\label{SEP}
\rho(1/3)^{\otimes k}=f_k(1/3)\Phi(1/3,k)+(1-f_k(1/3))\uno^{\otimes k}
\end{equation}is separable, as this state is the tensor product of separable states.

In addition, since $\lim_{p\rightarrow 1/3^+}\rho(p)^{\otimes k}=\rho(1/3)^{\otimes k}$, $\lim_{p\rightarrow 1/3^+}(1-f_k(p))=(1-f_k(1/3))$ and
\begin{align}\label{eq I}
\lim_{p\rightarrow 1/3^+}f_k(p)=f_k(1/3),
\end{align}it is clear that
\begin{align}\label{eq II}
\lim_{p\rightarrow 1/3^+}\Phi(p,k)=\Phi(1/3,k).
\end{align}

Now, Eq. (\ref{eq I}) and the fact that $f_k(1/3)\in (0,1)$ imply that there exists a $p_0>1/3$ such that for every $p\in (1/3, p_0]$ we have $$\frac{f_k(p)}{(n-1)-(n-2)f_k(p)}<f_k(1/3).$$In particular, $\frac{f_k(p_0)}{(n-1)-(n-2)f_k(p_0)}<f_k(1/3)$.

It is now easy to see that the state 
\begin{align}
&\frac{\frac{f_k(p_0)}{n-1}}{\frac{f_k(p_0)}{n-1}+(1-f_k(p_0))}\Phi(1/3,k)\nonumber\\&+\frac{1-f_k(p_0)}{\frac{f_k(p_0)}{n-1}+(1-f_k(p_0))}\uno^{\otimes k}\nonumber
\end{align}
is in the interior of the set of separable states. Indeed, this follows from Lemma \ref{lemma 1}, Eq. (\ref{SEP}) and the fact that
\begin{align}
&\frac{\frac{f_k(p_0)}{n-1}}{\frac{f_k(p_0)}{n-1}+(1-f_k(p_0))}\nonumber\\&=\frac{f_k(p_0)}{(n-1)-(n-2)f_k(p_0)}<f_k(1/3).\nonumber
\end{align}

Hence, using Eq. (\ref{eq II}) we can find a $\hat{p}\in (1/3, p_0]$ such that the state $$\frac{\frac{f_k(p_0)}{n-1}}{\frac{f_k(p_0)}{n-1}+(1-f_k(p_0))}\Phi(\hat{p},k)+\frac{1-f_k(p_0)}{\frac{f_k(p_0)}{n-1}+(1-f_k(p_0))}\uno^{\otimes k}$$is separable. Now, since $\hat{p}\in (1/3, p_0]$, Lemma \ref{lemma 2} guarantees that the state
$$\frac{\frac{f_k(\hat{p})}{n-1}}{\frac{f_k(\hat{p})}{n-1}+(1-f_k(\hat{p}))}\Phi(\hat{p},k)+\frac{1-f_k(\hat{p})}{\frac{f_k(\hat{p})}{n-1}+(1-f_k(\hat{p}))}\uno^{\otimes k}$$is also separable.

This concludes the proof.
\end{proof}

\section{Conclusions}

In this work we have studied GME in the multiple-copy scenario. We have provided a full characterization of the set of GME-activatable states as those that are not partially separable (Theorem 1). Furthermore, we have shown in Theorem 2 that this equivalence requires the asymptotic limit of infinitely many copies: given any $k\in\mathbb{N}$, we have explicitly constructed GME-activatable $n$-partite states of a fixed local dimension for all $n\geq3$ such that $k$ copies of them remain biseparable. These two results were actually conjectured in the recent work \cite{GMEactivation}. Nevertheless, while Theorem 1 completely settles conjecture (ii) in \cite{GMEactivation}, Theorem 2 gives only a partial answer to conjecture (i) in that paper, which stated that for every $k\geq 2$ there exists a GME-activatable state $\rho$ such that $\rho^{\otimes(k-1)}$ is biseparable but $\rho^{\otimes k}$ is GME. Theorem 2 proves the essential part of the conjecture (that is, the activation of GME in general requires an unbounded number of copies) and, actually, this implies that the aforementioned property conjectured in \cite{GMEactivation} has to hold for infinitely many different values of $k$. However, this does not rule out that for some specific values of $k$ the set of GME-activatable states with $k-1$ copies could be equal to the set of GME-activatable states with $k$ copies. Thus, this question remains open.

The construction for Theorem 2 uses PEN states, which we have recently introduced in \cite{CPV21}. An $n$-partite PEN state $\sigma_{G,\Upsilon}$ is defined by an undirected graph $G=(V,E)$ with vertices $V=[n]$ and edges $E\subseteq\{(i,j):i,j\in V, i<j\}$ and, given $E$, a set of bipartite states $\Upsilon=\{\rho(i,j)\}_{(i,j)\in E}$ so that
\begin{equation}\label{pen}
\sigma_{G,\Upsilon}=\bigotimes_{(i,j)\in E}\rho_{ij}(i,j).
\end{equation}
Therefore, PEN states are those that can be generated by only distributing bipartite entanglement among different pairs of parties. Thus, they constitute a realistic class of multipartite states that are relatively simple to prepare, underlying the current investigations on quantum networks as platforms for quantum information processing \cite{reviewn}. Moreover, their convenient mathematical structure makes it possible to exploit the much better developed theory of bipartite entanglement in order to analyze the complex structure of entanglement in the multipartite scenario, as the proof of Theorem 2 further exemplifies. It was shown in \cite{CPV21} that the mere distribution of bipartite entanglement in networks does not guarantee GME, i.e.\ there exist PEN states in which $G$ is connected and all states in $\Upsilon$ are entangled but such that $\sigma_{G,\Upsilon}$ is biseparable. It is worth mentioning in this context that our results here imply that the above property breaks down if the nodes share sufficiently many copies of the states in $\Upsilon$. Since, under the above assumptions on $G$ and $\Upsilon$, $\sigma_{G,\Upsilon}$ is not partially separable, Theorem 1 implies that the PEN state is GME-activatable and, hence, there always exists a value of $k\in\mathbb{N}$ such that $\sigma_{G,\Upsilon}^{\otimes k}=\bigotimes_{(i,j)\in E}[\rho_{ij}(i,j)]^{\otimes k}$ is GME. In \cite{CPV21} we had already made the weaker observation that all PEN states such that $G$ is connected and the states in $\Upsilon$ are distillable must be GME-activatable. We mention in passing that, reversing this implication, if there existed a PEN state that was not GME-activatable, this would mean that at least one state in $\Upsilon$ is not distillable, which could have been used to tackle the long-standing open question of the existence of bound entanglement beyond states which are positive under partial transposition. However, Theorem 1 closes this path.

From a more general perspective, in the line of \cite{GNME1} that introduces the notion of genuine network multipartite entanglement (GNME), Theorem 1 exposes further the fundamental limitations of GME as an experimental test of truly multipartite entanglement when multiple copies can be controlled. However, contrary to the set of biseparable states, the set of non-GNME is not closed under LOCC manipulation and relevant resources in this paradigm do not pass the GNME test. This is the case, for instance, of PEN states that underpin protocols held in quantum networks \cite{reviewn} (such as the star network in which a powerful central laboratory prepares entangled states for satellite nodes \cite{Var19}). Indeed, if all parties share sufficient bipartite entanglement as to enable perfect teleportation any state of any given local dimension can be obtained from them and they are, therefore, universal resources in the standard paradigm of state manipulation under LOCC. Hence, from a more practical point of view, GME can be seen as a benchmark to produce multipartite quantum states useful for applications (possibly conditioned on further LOCC postprocessing) in the light of \cite{gmesensing1, gmesensing2,GMEkey}. In fact, from this perspective our result uncovers the possibility in general to produce useful multipartite entanglement by mixing partially separable states, as the work of \cite{biseparableqcka} exemplifies in a particular case in the context of quantum cryptography. Looking at GME activation as a resource to be exploited led us to consider the question of whether there is a general upper bound for $k$ such that $\rho^{\otimes k}$ is GME given that $\rho$ is GME-activatable. This is answered in the negative in Theorem 2, implying that the preparation of GME states by mixing partially separable states can require control over an arbitrarily large number of copies. In this context, it would be then interesting to find general means to assess what the minimal number of copies necessary to display GME is for a given non-partially-separable state. We leave this problem for future research.

\begin{acknowledgments}
We thank Ludovico Lami for useful comments. This research was funded by the Spanish Ministerio de Ciencia e Innovaci\'{o}n (grant PID2020-113523GB-I00) and Comunidad de Madrid (grant QUITEMAD-CMS2018/TCS-4342). C.P. is partially supported as well by the Spanish Ministerio de Ciencia e Innovaci\'{o}n  (grant CEX2019-000904-S funded by MCINN/AEI/10.13039/501100011033). JIdV also acknowledges financial support from Comunidad de Madrid (Multiannual Agreement with UC3M in the line of Excellence of University Professors EPUC3M23 in the context of the V PRICIT).

\end{acknowledgments}

\end{document}